\documentclass[10pt]{article}
\usepackage[T1]{fontenc}
\usepackage{amsmath}
\usepackage{amsthm}
\usepackage{amssymb}
\usepackage{commath}
\usepackage{multicol}
\usepackage{url}
\usepackage{xparse}
\NewDocumentCommand{\ceil}{s O{} m}{%
  \IfBooleanTF{#1} 
    {\left\lceil#3\right\rceil} 
    {#2\lceil#3#2\rceil} 
}
\usepackage{enumerate}
\usepackage[affil-it]{authblk}

\newcommand{\Z}{{\mathbb Z}}
\newcommand{\N}{{\mathbb N}}
\newcommand{\C}{{\mathbb C}}

\newcommand{\CC}{{\cal C}}
\newcommand{\NN}{{\cal N}}

\newtheorem{definition}{Definition}
\newtheorem{proposition}{Proposition}
\newtheorem{theorem}{Theorem}

\def\e{\,{\rm e}\,}
\def\P{|\Psi\rangle}

\DeclareMathOperator{\Tr}{Tr}
\def\floor#1{\left\lfloor #1\right\rfloor}
\makeindex
\title{On the Existence of
Absolutely Maximally Entangled
States of Minimal Support II}

\author{Antonio Bernal%
  \thanks{Electronic address: \texttt{abernal@ub.edu}\\
Supported by project FIS2013-41757-P  }}
\affil{Department de Matemàtiques i Informàtica. Universitat de Barcelona}

\begin{document}
\maketitle

\begin{abstract}
Absolutely maximally entangled, AME, states are pure multipartite states that give rise to the maximally mixed states when half or more of the parties are traced out. AME states find applications in fields like teleportation or quantum secret sharing, and the problem of finding conditions on their existence has been considered in a number of papers.

We consider here AME states of minimal support, that are simpler to analyse. 
An equivalence with coding theory gives a sufficient condition for their existence, that the number of sites be equal to the local dimension plus one, when the local dimension $d$ is a power of a prime number.

In this paper, an equivalence with latin hypercubes is used to prove that the above sufficient condition fails in the first case in which the local dimension is not a prime power, $d=6$. Results for other values of $d$ are also given.
\end{abstract}

\begin{multicols}{2}

\section{Introduction}

The states called \lq\lq $\floor{n/2}$-uniform'', were considered in \cite{scott} in the context of quantum error correcting codes, and later, they where called \lq\lq absolutely maximally entangled'', AME, states in \cite{qss}, in the context of quantum secret sharing schemes.

The properties of such states have been investigated in several papers, see \cite{bernal}, \cite{ame64}, \cite{qss} and \cite{ameexistenceandapplications}.

This paper deals with the problem of existence of $AME$ states that are supported in a minimal set of kets from the computational basis. We prove that the known existence result, valid when the local dimension is a prime power, fails in the first non prime power case, $d=6$. 

We also deal with the cases where $d=4$, $d=5$ and $d=7$.

In sections \ref{sec:amegeneral} and \ref{sec:mdscodes}, we give the basic definitions of general and minimally supported $AME$ states, and some known characterizations for the comodity of the reader.

In section \ref{sec:hypercubes}, we use a relationship between MDS codes and latin hypercubes to get a corresponding translation in terms of $AME$ states of minimal support and note that known properties of latin squares forbids certain $AME$ states of minimal support.

\section{Absolutely maximally entangled states}\label{sec:amegeneral}

We state directly the main definitions and a few basic properties, and refer the reader to \cite{bernal}, \cite{ame64}, \cite{qss}, \cite{ameexistenceandapplications}, and the references given there, for more explanations, proofs and details.

\begin{definition}
Let $n$ and $d$ be integers $n,d\ge 2$. Let $\P$ be a pure multipartite state in $n$ sites, where the local Hilbert space is $d$-dimensional. That is,
\[
\P\in (\C^d)^{\otimes n}.
\] 

We say that $\P$ is absolutely maximally entangled with $n$ sites and local dimension $d$, $AME(n,d)$, if for any partition of $\{1,\ldots,n\}$ into two disjoint subsets $A$ and $B$, with $|B|=m\le|A|=n-m$, the density obtained from $\P\langle\Psi|$ tracing out the sites on the entries in $A$ is multiple of the identity,
\[
\Tr_A\P\langle\Psi| = \frac{1}{d^{m}}Id_{\C^{\otimes m}}.
\]
\end{definition}

It is an open problem to find exactly for what values $n$ and $d$ do $AME(n,d)$ states exist. In the above references many partial results, both positive and negative, can be found.

Recall that, given a vector space $V$, a vector $v\in V$ and a basis ${\cal B}\subset V$, the support of $v$ in the basis $\cal B$ is the number of nonzero coordinates of $v$ in the basis $\cal B$.

A linear algebra argument gives a lower bound for the support of any AME state on the computational basis $|s_1,\ldots,s_n\rangle$, $0\le s_i\le d-1$, $1\le i\le n$.

\begin{proposition}
If $n$ and $d$ are integers, $n,d\ge 2$ and $\P$ is an $AME(n,d)$ state, then the support of $\P$ with respect to the computational basis is at least $d^{\floor{n/2}}$.
\end{proposition}

\begin{definition}
Given two integers $n$, $d$, with $n,d\ge 2$, we will say that an $AME(n,d)$ state $\P$ is of minimal support if the support of $\P$ in the computational basis is $d^{\floor{n/2}}$.
\end{definition}

The problem of finding $AME(n,d)$ states and that of finding $AME(n,d)$ states of minimal support are different ones. For example, in the above references, it is proved that $AME(6,2)$ states exist, but none of them is minimally supported.

It is readily seen that any $AME(n,d)$ state of minimal support can be expressed as $\P=\sum_{\nu=0}^{2^m-1}\e^{i\theta_\nu}|k_{\nu,1},\ldots,k_{\nu,n}\rangle$, for certain $(k_{\nu,1},\ldots,k_{\nu,n})$, $0\le k_{\nu,j}\le d-1$ and real phases $\theta_\nu$. It follows that if a state $\P$ is $AME(n,d)$ with minimal support for a trial of the phases $\theta_\nu$, the state obtainded from $\P$ with any other trial $\theta_\nu^\prime$ is also $AME(n,d)$ of minimal support. In particular, we can set $\theta_\nu=0$. Thus, when it comes to determine whether $AME(n,d)$ states of minimal support exist, only the set of tuples $(k_{\nu,1},\ldots,k_{\nu,n})$ matter.

\section{Characterization with MDS codes}\label{sec:mdscodes}

Let's briefly summarize some standard definitions of the theory of codes. See \cite{macwilliams-sloane} for more details.

A code $\CC$ over the alphabet $\Z_d=\{0,\ldots,d-1\}$ and wordlength $n$ is a subset of $\Z_d^n$. The Hamming distance of two words $w,w'\in\CC$, $D_H(w,w')$, is the number of coordinates on wich $w$ and $w'$ differ. Denote by $\delta\in\Z^+$ the minimum of $D_H(w,w')$, where $w$ and $w'$ are different words of the code. We call $\delta$ the minimum distance of $\CC$. The well known Singleton bound establishes that $|\CC|\le d^{n-\delta+1}$. A code is called \lq\lq maximun distance separable'', MDS, if the singleton bound is an equality. 

In this case, if we define the positive integer $k=n-\delta +1$, we have that $|\CC|=d^k$. We call $k$ the dimension of $\CC$. We can call $k$ the combinatorial dimension of $\CC$, to stress the fact that $\CC$ might have no particular algebraic structure for a general integer $d$.

The following property follows from the characterization of MDS codes in terms of orthogonal arrays, and can be found in \cite{huntemann} or \cite{macwilliams-sloane}.

\begin{proposition}\label{th:codeminusone}
If there is an MDS code, not necessarily linear, of length \(n\) and dimension \(k\), there is also an MDS code of length \(n-1\) and dimension \(k\). 
\end{proposition}

The following characterization has been proved  in \cite{ame64} and \cite{ameexistenceandapplications}.

\begin{theorem}\label{th:codes}
Given two integers $n, d\ge 2$, an $AME(n,d)$ state of minimal support exists if, and only if, there is and MDS code over $\Z_d$, of wordlength $n$, and minimum distance $\delta = \ceil{n/2}+1$, equivalently $k=\floor{n/2}$. The words in the code and the kets in the state are in one onto one correspondence.
\end{theorem}

For any given local dimension $d$, there are always $AME(n=3,d)$ states of minimal support, just consider $d^{-1/2}(|0\,0\,0\rangle+\cdots+|d-1\,d-1\, d-1\rangle)$. The set of all integers $n\ge 2$ such that an $AME(n,d)$ state of minimal support exists is therefore nonempty and, in fact, it is an interval, see \cite{bernal}, and the references given there.

\begin{proposition}
For any integer $d\ge 2$, there is an integer $\NN(d)$ such that, an $AME(n,d)$ state of minimal support exists if, and only if, $n\le \NN(d)$.
\end{proposition}

On the case where $d$ is a prime power, the alphabet $\{0,\ldots,d-1\}$ can be given a unique field structure, $GF(d)$. Then, we can resort to the theory of Red Solomon codes and their extensions, that are known to be MDS. We have then the following existence result, see \cite{bernal}, \cite{ameexistenceandapplications} and the references therein.

\begin{theorem}\label{th:primepowercase}
Let $d\ge 3$ be an integer that is a power of a prime number. Then, there is an $AME(d+1,d)$ state of minimal support. Thus $\NN(d)\ge d+1$.
\end{theorem}

\section{Characterization with latin hypercubes}\label{sec:hypercubes}

\begin{definition}
	Given two integers $k, d\ge 2$, a latin $k$-hypercube of order $d$ is a $k$-dimensional array of integer numbers
	\begin{equation}\label{eq:hypercube}
	L_{j_1,\ldots,j_k},
	\end{equation}
	$0\le j_\nu\le d-1$, such that $0\le L_{j_1,\ldots,j_k}\le d-1$ and that fixing all indices minus any one of them, the $d$ resulting integers (\ref{eq:hypercube}) are all diferent (and thus are all $\{0,\ldots,d-1\}$).
\end{definition}

Latin hypercubes are a generalisation of latin squares, the case where $k=2$. If in a latin hypercube, all but two of the indices are fixed, a latin square results. Latin squares and hypercubes have many applications in several areas, like combinatorial design of experiments. There are a number of nontrivial questions concerning the conditions on the existence of those structures. See, for example, \cite{dougherty}.

\begin{definition}
	Two latin squares of order $d$ $L$ and $M$ are said to be orthogonal if the $d^2$ pairs $(L_{ij},M_{ij})$ are all different. To latin $k$-hypercubes $L_{j_1,\ldots,j_k}$ and $L^\prime_{j_1,\ldots,j_k}$ are said to be orthogonal if, fixing all but two of the indices, the resulting latin squares are orthogonal. Several latin $k$-hypercubes are mutually orthogonal if any two of them are orthogonal.
\end{definition}

Latin hypercubes are connected to MDS codes by the following property \cite{dougherty}.

\begin{theorem}\label{th:cubes}
	Given $d,n\ge 2$ integers, there exists an MDS code over $\Z_d$ with wordlength $n$ and minimum distance $\delta$ if, and only if, there are $\delta-1$ mutually orthogonal latin $k$-hypercubes ($k=n-\delta+1$) of order $d$.
\end{theorem}

Theorems \ref{th:codes} and \ref{th:cubes} can now be read together, to have the following characterization:

\begin{theorem}
	Given two integers $n, d\ge 2$, an $AME(n,d)$ state of minimal support exists if, and only if, there are $\ceil{\frac{n}{2}}$ mutually orthogonal latin $\floor{\frac{n}{2}}$-hypercubes of order $d$.
\end{theorem}

We know, by proposition \ref{th:primepowercase}, that, when $d$ is a power of a prime number, there are $AME(d+1,d)$ states of minimal support. Thus the question of whether, given and integer $d\ge 3$, there are $\ceil{\frac{d+1}{2}}$ mutually orthogonal latin $\floor{\frac{d+1}{2}}$-hypercubes of order $d$, has a positive answer in the case where $d$ is the power of a prime number. It turns out that the answer is negative in the first non-prime power case, $d=6$.

Indeed, let's consider the case where $n=4$ and $d=6$. In the above characterization we have $\ceil{n/2}=\floor{n/2}=2$, so the existence of an $AME(4,6)$ state of minimal support is equivalent to the existence of two orthogonal latin squares of order $d=6$. 

However, there are no couples of orthogonal latin squares of order 6. 

The non existence of a couple of orthogonal latin squares of any order $d\equiv 2 (\textrm{mod}\, 4)$ was conjectured in 1782 by Euler \cite{euler}. The proof in the case where $d=2$ is straightforward. The case where $d=6$ was proved in 1901, \cite{tarry-1,tarry-2}. In the remaining cases, Euler's conjecture has been proved to be false \cite{bose-parker-shirkhande}.

\begin{theorem}\label{th:case-4-6}
	There is no $AME(4,6)$ state of minimal support and $\NN(6)=3$.
\end{theorem}
\begin{proof}
	We have seen that $\NN(6)\le 3$. As we noted earlier, $\NN(d)\ge 3$ for any $d\in\N$, $d\ge 2$.
\end{proof}

\section{The MDS conjecture}\label{sec:mdsconjecture}

For any integers $k$ and $d$, we define $M(k,d)$ to be the maximum length of the MDS codes of combinatorial dimension $k$ over an alphabet of size $d$. If there are MDS codes of any length of dimension $k$ and alphabet size $d$, we define $M(k,d)=\infty$. 

If $d\ge 2$ is a power of a prime number, we define $L(k,d)$ as the maximum length of the \emph{linear} MDS codes over the field $GF(d)$ of dimension $k$ over $GF(d)$. In this case, $L(k,d)\le M(k,d)$.

The \emph{MDS conjecture} states that, for $d$ power of a prime number, $L(k,d)=d+2$, when $d=2^j$, and $k\in\{3,d-1 \}$, and $L(k,d)=d+1$ otherwise. This conjecture has been extensively studied and it is known to hold in many particular cases. See \cite{huntemann} for a detailed account of known results. An important recent advance is that the MDS conjecture is true when $d$ is prime, see \cite{simeonball}.

For $d$ power of a prime, we can call \emph{general MDS conjecture} the same statement as above with $L$ replaced by $M$. See \cite{huntemann} for the definition of the general MDS conjecture when $d$ is not a prime power and known results about it.

\section{Several more results}

	We already know that $\NN(2)=3$, since there is no $AME(4,2)$ state, even non-minimally supported, \cite{higuchisudbery}, $\NN(3)=4$, since $\NN(3)\le 2\cdot 3-2$, \cite{bernal}.

By theorem \ref{th:primepowercase}, $\NN(d)\ge d+1$ if $d$ is a prime power. The following result shows that we don't have an equality in all cases.

\begin{theorem}
	$\NN(4)=6$, $\NN(5)=6$, $\NN(7)=8$.
\end{theorem}

\begin{proof}
	Suppose that $d=4$. Pick $k=3$. It is known that, for $k=3$ and $d$ an even prime power, $L(k=3,d) = d+2$ \cite{huntemann,macwilliams-sloane}. If $n=6$, then $\floor{n/2}=3=k$, so we have the right dimension to apply theorem \ref{th:codes} and we have a minimally supported AME state with $n=6$ and $k=4$ given by a linear code over $GF(4)$. Then $\NN(4)\ge 6$. Conversely, we have that $\NN(4)\le 2d-2=6$, \cite{bernal}.
	
	For the cases where $d=5$ or $d=7$, we use a result proved in \cite{kokkalaetal}, that for $d\in\{5,7 \}$, $k\ge 2$ and $\delta=n-k+1\ge 3$, there exists a, not necessarily linear, MDS code of length $n$ and combinatorial dimension $k$ over $\Z_d$ if, and only if, $n\le d+1$. This is like the general MDS conjecture for $d=5$ or $d=7$, but for codes with minimum distance $\delta \ge 3$.
	
	This result of \cite{kokkalaetal} follows by proving that the general MDS codes, with the appropiate minimum distance, can be obtained from linear MDS codes with the same parameters by a permutation of the code coordinates and a permutation of symbols independently in each coordinate.
		
	Applying the result to $d=5$ and $k=3$ we see that there is no $AME(7,5)$ state of minimal support, as noted in \cite{bernal}; Applying it to $d=7$ and $k=4$ we see that there is no $AME(9,7)$ state of minimal support.
\end{proof}

We can say something about other cases, conditionally to the validity of the general MDS conjecture, as stated in section \ref{sec:mdsconjecture}.

\begin{proposition}
	Let $d \ge 8$ be a power of a prime number. Suppose that the general MDS conjecture holds for the case where the alphabet size is $d$ and the dimension is  $\floor{\frac{d+2}{2}}$. Then $\NN(d)=d+1$.
\end{proposition}
\begin{proof}
If $d\ge 8$, consider 
	\[k=\floor{\frac{d+2}{2}}\ge 5\] 
	
	Then $3 < k < d-1$. So $k\notin\{3,d-1\}$ and, by the general MDS conjecture, $M(k,d)=d+1$. So there is no MDS code over $GF(d)$ of length $d+2$ and combinatorial dimension $k=\floor{(d+2)/2}$, so $\NN(d)\le d+1$.
\end{proof}

The survey \cite{huntemann} gives a table of known results on the maximum length \(n\) of MDS codes with a given alphabet size \(d\le 100 \) and several values of the dimension \(k\). Considering that, when looking for AME states with minimal support, we are interested in codes with $k=\floor{n/2}$ we can use that reference for getting information on $\NN(d)$ for more values of $d$.

Just as an example, consider the case where \(d=10\). In \cite{huntemann}, table 5.1, we find that \(M(k=8,d=10)\le 11\). Then, we know that there is no general MDS code of length \(n=12\) and combinatorial dimension \(k=8\). From proposition \ref{th:codeminusone}, it follows that there is no MDS code of length \(n=16\) and dimension \(k=\floor{n/2}=8\). By theorem \ref{th:codes}, there is no $AME(16,10)$ state of minimal support. Then \(\NN(10)\le 15 \), which is a better bound than \(\NN(10)\le 2\cdot 10-2= 18\), given in \cite{bernal}. That kind of estimates, however seem far to be tight; the number \( \NN(d)\) seems to be much closer to $d+1$ than to $2d-2$ or $2d-3$.

The case where $d=6$ considered in theorem \ref{th:case-4-6}, that $\NN(6)=4$, can be obtained from \cite{huntemann}, table 5.1 too, it is stated there that $M(k=2,d=6)=3$, that translates in $\NN(6)\le 3$. The result is derived from a packing problem on vector spaces over Galois Fields \cite{manerisilverman}, a subject that bears also an indirect relation with Euler's conjecture as used here.

\section{Conclusion}

The theory of $AME$ states of minimal support bear relation with diverse areas of Mathematics such as MDS codes, orthogonal arrays, latin hypercubes or finite geometries.

All of those areas are known to be related with each other and have long standing problems, like the MDS conjecture, that, when solved, will throw light on the existence problem we deal with in this paper.

The author wishes to thank D. Krotov and S.T. Dougherty for useful conversations and remarks.

\end{multicols}

\end{document}